\documentclass[12pt]{extarticle}
\usepackage{extsizes}
\usepackage[reqno]{amsmath}
\usepackage{amsmath}
\usepackage{amssymb}
\usepackage{amsthm}
\usepackage{amscd}
\usepackage{amsfonts}
\usepackage{mathtools}
\usepackage{graphicx}
\usepackage{dsfont}
\usepackage{fancyhdr}
\usepackage{enumerate}
\usepackage{youngtab}
\usepackage[utf8]{inputenc}
\usepackage{comment}
\usepackage{mathdots}
\usepackage{subfigure}
\usepackage[margin=1.5in]{geometry}
\usepackage{graphicx}
\usepackage{algorithm}
\usepackage{algpseudocode}
\usepackage{color}
\usepackage{hyperref}
\usepackage{notation}
\usepackage{tikz}
\begin{document}
\AtEndDocument{%
  \par
  \medskip
  \begin{tabular}{@{}l@{}}%
    \textsc{Gabriel Coutinho}\\
    \textsc{Dept. of Computer Science} \\ 
    \textsc{Universidade Federal de Minas Gerais, Brazil} \\
    \textit{E-mail address}: \texttt{gabriel@dcc.ufmg.br} \\ \ \\
    \textsc{Frederico Cançado} \\
    \textsc{Dept. of Mathematics} \\ 
    \textsc{Universidade Federal de Minas Gerais, Brazil} \\
    \textit{E-mail address}: 
    \texttt{fredericocancado@ufmg.br}\\ \ \\
    \textsc{Thomás Jung Spier} \\
    \textsc{Dept. of Combinatorics and Optimization} \\ 
    \textsc{University of Waterloo, Canada} \\
    \textit{E-mail address}: \texttt{tjungspier@uwaterloo.ca}
  \end{tabular}}

\title{Orthogonal polynomials, quantum walks and the Prouhet–Tarry–Escott problem}
\author{Frederico Cançado \and Gabriel Coutinho\footnote{gabriel@dcc.ufmg.br --- remaining affiliations in the end of the manuscript.} \and Thomás Jung Spier}
\date{\today}
\maketitle
\author
\vspace{-0.8cm}

\begin{abstract} 
    This paper is motivated by the following problem. Define a quantum walk on a positively weighted path (linear chain). Can the weights be tuned so that perfect state transfer occurs between the first vertex and any other position? We do not fully answer this question --- in fact, we show that a particular case of this problem is equivalent to a solution of a particular case of the well known Prouhet–Tarry–Escott problem, deeming our original task certainly harder than anticipated. In our journey, we prove new results about sequences of orthogonal polynomials satisfying three-term recurrences. In particular, we provide a full characterization of when two polynomials belong to such a sequence, which (as far as we were able to ascertain) was known only for when their degrees differ by one. 
\end{abstract}

\begin{center}
\textbf{Keywords}
perfect state transfer ; strong cospectrality ; orthogonal polynomials ; Prouhet–Tarry–Escott problem.
\end{center}


\section{Introduction}\label{sec:introduction}

In this paper we are concerned with a chain of qubits subject to a nearest neighbor interaction regime. The coupling strengths must be positive, and we allow for a possible external magnetic field at each qubit. In mathematical terms, we are modeling the qubit network by an undirected, weighted graph, with positive edge weights, possibly with weighted loops at each vertex, whose adjacency matrix is tridiagonal. The main problem we study is that of completely transferring the state at the excited site to another site in the network, a phenomenon known as ``perfect state transfer". According to Schrödinger's equation, this is equivalent to the following mathematical statement:

\begin{itemize}
	\item If $A$ is symmetric and tridiagonal of dimension $d+1$, with positive off-diagonal terms, when is it the case that for some value of $t$ an off-diagonal entry of $\exp(\ii t A)$ has absolute value equal to $1$?
\end{itemize}

Since the seminal work in~\cite{ChristandlPSTQuantumSpinNet2}, there has been a long list of works devoted to the study of perfect state transfer in all sorts of graphs. We refer to the works in~\cite{KayPerfectcommunquantumnetworks,CoutinhoGodsilSurvey} and references therein for detailed introductions to the topic.

The connection between tridiagonal matrices and sequences of orthogonal polynomials is well known, and in the context of quantum walks it has been successfully exploited (see for instance~\cite{VinetZhedanovHowTo, vinet2012dual, VinetZhedanovAlmost}). In~\cite{VinetZhedanovHowTo} a complete characterization of perfect state transfer between extremal sites of weighted paths was obtained. These examples also admit state transfer between any pair of symmetric positions, though they do not provide in this latter case an exhaustive list. What, however, seems harder and somewhat unintuitive is the problem of characterizing perfect state transfer between asymmetric positions. In~\cite{Coutinho2019a}, an infinite family of weighted paths admitting perfect state transfer between asymmetric sites was obtained, also by exploring the connection to orthogonal polynomials. 

We raise the question of whether it is possible to completely describe all possible cases of perfect state transfer in weighted chains. While this still seems out of reach, in this paper we present several significant advances towards this quest. Again, assume the chain contains $d+1$ sites, and all coupling strengths are positive.

Our first main result is a characterization of when two monic polynomials with distinct real roots belong to an orthogonal polynomial sequence corresponding to a tridiagonal matrix~\footnote{It is plausible this result is already known, however our best efforts could not locate a reference.} --- which should be of interest on its own.

\begin{theorem*}
Let $q$ and $p$ be two monic polynomials with distinct real zeros of degrees $\ell$ and $m$, respectively, where $\ell < m$. Then there exists a tridiagonal matrix with positive off-diagonal entries determining a sequence of orthogonal polynomials containing $p$ and $q$, if and only if, the following property holds:
\begin{itemize}
	\item[] In any open interval determined by two roots of $q$, there is a root of $p$, and all the zeros of $q$ lie in the open interval determined by the smallest and largest zeros of $p$. 
\end{itemize}
\end{theorem*}

This theorem is restated as Theorem~\ref{thm:construction_OPS}, and proved in Section~\ref{sec:orthogonal_polynomial_sequence}.

When studying perfect state transfer, an important concept is that of cospectral vertices. Loosely speaking, two vertices are cospectral if quantum walks~\footnote{This also applies to classical random walks.} of any duration and starting at each vertex have the same probability of return, for both vertices. This concept has been studied in classical contexts since the 1970s, and its connection to quantum walks was studied in \cite{godsil2017strongly}. In particular, perfect state transfer between two vertices is possible only if they are cospectral.

Our second main result is a characterization of which positions in a weighted path might admit a pair of cospectral vertices. As usual, we assume the edges connecting neighbouring vertices have positive weight.

\begin{theorem*}
	Let $0\leq \ell<m\leq d$. There exists a weighted path with $d+1$ vertices and positive edge weights in which (vertices in positions) $\ell$ and $m$ are cospectral if, and only if, $\ell<\frac{d}{2}<m$.
\end{theorem*}

We remark that this result should also have implications in the study of classical random walks and Markov chains. We restate this as Theorem~\ref{thm:cospectral_characterization} and prove it in Section~\ref{sec:cospectral_vertices_positions}.

Perfect state transfer has been characterized between the extremal positions of linear chains, and there are known infinite families admitting state transfer between the first and the second-to-last positions. In light of the result above, it is natural to ask if perfect state transfer might happen between vertices at the first position and at the closest possible, which will be right after the middle of the chain. If that is the case, then the two vertices are cospectral and periodic, as we will soon explain. Our third main result shows that such examples of pairs of vertices will imply the existence of ideal solutions of the Prouhet--Tarry-Escott problem, which seems to have been first proposed in the 1750s and for which solutions are known only for a few cases~\footnote{Cases corresponding to $n$ integers with $3 \leq n \leq 10$, and for $n = 12$.} (see \cite{coppersmith2024ideal} for a recent exposition).

\begin{theorem*}
	In a linear chain with vertices $\{0,\cdots,d\}$, if there are periodic and cospectral vertices in positions $0$ and $n = \left\lceil\frac{d+1}{2}\right\rceil$, then there is an ideal solution to the Prouhet--Tarry--Escott problem on $n$ integers.
\end{theorem*}

A slightly more general version of this theorem is stated as Theorem~\ref{thm:PTE_equivalence} and proved in Section~\ref{sec:periodic_cospectral_vertices}.

Finally, we finish the paper discussing some consequences to the study of perfect state transfer in linear chains in Section~\ref{sec:PST}, along with some open questions in the end.


\subsection{A quick recap of orthogonal polynomials and tridiagonal matrices}\label{sec:preliminaries1}

In this section we briefly introduce the notation we use throughout the paper, while recalling the main results and definitions from the literature. We refer the reader to~\cite{chihara2011introduction} for a complete introduction to orthogonal polynomials, and to~\cite{VinetZhedanovHowTo} for an elementary introduction suited for an application to quantum walks.

Consider a spin chain on $d+1$ sites with Jacobi matrix 
\begin{equation}\label{eq:jacobi_matrix}J := \left[ {\begin{array}{cccc}
	a_0 & \lambda_1 & & \\
    \lambda_1 & a_1 & \ddots & \\
	   & \ddots & \ddots & \lambda_d \\ 
	   &  &  \lambda_d & a_d \\ 
	\end{array} } \right],
\end{equation}
where we assume that $a_i,\lambda_i \in \mathbb{R}$ and $\lambda_i > 0$ for all $i$. In this paper, we shall refer to these as \textbf{$\mathbf{d}$--chains} (with $d+1$ sites).

Its eigenvalues are simple \cite{chihara2011introduction}. We shall adhere to the convention that they are ordered largest to smallest, by $\theta_0 > \cdots > \theta_d$. Observe that $J$ is the adjacency matrix of a weighted path $P$, shown in Figure~\ref{fig:wightedpath}.

\begin{figure}[H]
    \centering
\begin{tikzpicture}
\draw (0,0) node[draw, circle](q1){};
\draw (2,0) node[draw, circle](q2){};
\draw (4,0) node[draw, circle](q3){};
\draw (6,0) node[](qd){$\dots$};
\draw (8,0) node[draw, circle](qn){};

\draw[-] (q2) edge node[below, midway, fill=white]{$\lambda_1$} (q1);
\draw[-] (q3) edge node[below, fill=white]{$\lambda_2$} (q2);
\draw[-] (qd) edge node[below, fill=white]{$\lambda_3$} (q3);
\draw[-] (qn) edge node[below, fill=white]{$\lambda_d$} (qd);
\draw[->] (q1) edge[loop above] node[above,midway, fill=white]{$a_0$} (q1);
\draw[->] (q2) edge[loop above] node[above,midway, fill=white]{$a_1$} (q2);
\draw[->] (q3) edge[loop above] node[above,midway, fill=white]{$a_2$} (q3);
\draw[->] (qn) edge[loop above] node[above,midway, fill=white]{$a_d$} (qn);
\end{tikzpicture}
    \caption{Weighted path corresponding to tridiagonal matrix, or a $d$--chain.}
    \label{fig:wightedpath}
\end{figure}
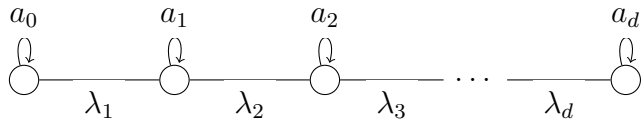

Let $p_{-1}\equiv 0$, $p_0 \equiv 1$, $p_1(x):=x-a_0$ and $p_{k+1}(x):=(x-a_k)p_k-\lambda_k^2 p_{k-1}(x)$ for $k \in [d]$ be the orthogonal polynomial sequence associated with $J$. In this case, $p_k(x)$ has degree $k$ and is the characteristic polynomial of the principal submatrix of $J$ formed by the first $k$ rows and columns. Therefore,
\begin{align}
 p_{d+1}=(x-\theta_d)\cdots(x-\theta_0).   \label{eq:charpol}
\end{align}

By rearranging the recurrence relation defining the orthogonal polynomial sequence and observing that $p_{d+1}(\theta_s)=0$, it also follows that

\[\ket{\theta_s}=\pmat{ p_0(\theta_s) \\[5pt] \dfrac{p_1(\theta_s)}{\lambda_1} \\ \vdots \\ \dfrac{p_d(\theta_s)}{\lambda_1\cdots \lambda_d}}\]

\noindent is the eigenvector of $J$ corresponding to the eigenvalue $\theta_s$. Note that, $\{\ket{\theta_s}\}_{s=0,\dots, d}$ is an orthogonal basis of eigenvectors of $J$.

\subsection{Quantum walks and cospectral vertices}\label{sec:preliminaries2}

We say perfect state transfer occurs from site $\ell$ to site $m$ if, for some real number $\tau$, 
\begin{align}
    \exp(\ii \tau J) \ket \ell = \gamma \ket m, \label{eq:pst}
\end{align}
where $\gamma$ is a complex number in the unit circle, and $\ket{k}$ is the vector with zero on all positions except for the $k^{th}$ position, which has a one.

It is straightforward to verify that if there is perfect state transfer from $\ell$ to $m$, then there is also from $m$ to $\ell$ at the same time (and with the same $\gamma$). Because we allow $J$ to be weighted, it is without loss of generality that we can take $\tau = \pi$, since we can just rescale the matrix by an overall multiplicative factor, and $\gamma = 1$, since we can always add a multiple of the identity matrix $I$ to $J$. Thus, Equation~\eqref{eq:pst} is, up to rescaling and translation, equivalent to having:
\begin{align}
\forall s,\ \mathrm{e}^{\ii \pi \theta_s} \braket{\theta_s}{\ell} = \braket{\theta_s}{m}. \label{eq:pst2}
\end{align}

If $\ell = 0$ and $m = d$, a complete and useful description of all solutions to this problem was obtained in~\cite{KayReviewPST} with a direct method, and in~\cite{VinetZhedanovHowTo} applying orthogonal polynomials. For the case $\ell = 0$ and $m = d-1$, examples were constructed in~\cite{KayPerfectcommunquantumnetworks} and~\cite{coutinho2019perfect}.

Using the orthogonal polynomial sequence we can write out conditions depending only on the polynomials and the eigenvalues that are equivalent to perfect state transfer between any two vertices in a $d$--chain.

\begin{theorem}\label{thm:PST_characterization_paths} Consider a $d$--chain. There is perfect state transfer between vertices in positions $\ell$ and $m$ if and only if there exists a scaled and translated transformed $d$--chain corresponding to polynomials $\{p_r\}_{r=0}^{d+1}$, with integer eigenvalues $\{\theta_s\}_{s=0}^d$, so that
\begin{itemize}
    \item For some $C>0$ we have that $p_\ell(\theta_s)=(-1)^{\theta_0-\theta_s}Cp_m(\theta_s)$,
for each eigenvalue $\theta_s$.
\end{itemize}
This transformed $d$--chain admits perfect state transfer at time $\pi$.
\end{theorem}

While this theorem is implicit in previous works (see \cite{VinetZhedanovHowTo} for instance), we will provide a proof shortly.

An important concept in the study of quantum walks is that of cospectral vertices.~\footnote{For graphs that admit eigenvalues with higher multiplicity, one must distinguish between cospectral and strongly cospectral vertices. We refer to \cite{godsil2017strongly} for a complete treatment. When the eigenvalues are simple, both definitions are equivalent, so we may as well adopt a simplified approach.}. We say two vertices $\ell$ and $m$ are \textit{cospectral} in a $d$--chain if
\[
    \forall s,\ |\braket{\theta_s}{\ell}| = |\braket{\theta_s}{m}|.
\]

It follows that the two vertices are cospectral if, and only if,  

\[\forall s,\ \left|\dfrac{p_\ell(\theta_s)}{\lambda_1\cdots \lambda_\ell}\right| = \left|\dfrac{p_m(\theta_s)}{\lambda_1\cdots \lambda_m}\right|.\]

We can do slightly better:

\begin{lemma}\label{lem:orthonormality}
    The vertices $\ell$ and $m$ of a $d$--chain are cospectral if, and only if, for some constant $C>0$, $|p_\ell(\theta_s)|=C|p_m(\theta_s)|$ for every $s=0,\dots,d$.
\end{lemma}
\begin{proof} The necessity of the statement is clear, so let us prove the sufficiency. Consider the matrix $U$ with entries 
\[
U_{\ell,s}=\frac{p_\ell(\theta_s)}{\lambda_1\cdots\lambda_\ell u_s},
\]
\noindent where $u_s$ is the norm of $|\theta_s\rangle$. Note that the column vectors of $U$ form the orthonormal basis of eigenvectors of $J$. Therefore $U$ is orthogonal, and, as a consequence, its rows also form an orthonormal basis of $\mathbb{R}^{d+1}$. The hypothesis that $|p_\ell(\theta_s)|=C|p_m(\theta_s)|$ for every $s$ and some constant $C>0$ implies that the $\ell$-th and $m$--th rows of $U$ with squared coordinates are multiple of each other. Since the $\ell$-th and $m$--th rows of $U$ both have norm $1$, it follows that the squared coordinates of these rows are in fact equal. In other words, we have
\[
\left|\dfrac{p_\ell(\theta_s)}{\lambda_1\cdots \lambda_\ell u_s}\right| = \left|\dfrac{p_m(\theta_s)}{\lambda_1\cdots \lambda_m u_s}\right|,\qquad s=0,\dots,d,
\]
\noindent and hence $\ell$ and $m$ are cospectral.
\end{proof}

\begin{proof}[Proof of Theorem~\ref{thm:PST_characterization_paths}] 
    Equation \eqref{eq:pst2} is equivalent to having both vertices cospectral, and for all $\theta_i$ in the support of $\ell$ (and $m$), $\theta_i$ even if the signs of $p_\ell(\theta_s)$ and $p_m(\theta_s)$ agree, and odd otherwise. From Lemma~\ref{lem:orthonormality}, this is equivalent to the stated condition.
\end{proof}

A key technique to study cospectral vertices consists of examining the characteristic polynomial of vertex-deleted subgraphs. If $P$ is a weighted path, we denote by $\phi^P$ the characteristic polynomial of the adjacency matrix of $P$. For instance, if $P$ is as in Figure~\ref{fig:wightedpath}, with adjacency matrix $J$ as in \eqref{eq:jacobi_matrix}, then $\phi^P = p_{d+1}$ from \eqref{eq:charpol}.

We let $P \setminus \ell$ denote the graph (which may now be disconnected) obtained upon deleting $\ell$ (and its incident edges) from $P$. It is well-known that two vertices $\ell$ and $m$ are cospectral in $P$ if and only if $\phi^{P\setminus \ell}=\phi^{P\setminus m}$ (see \cite[Theorem 3.1]{godsil2017strongly} or \cite[Chapter 4]{GodsilAlgebraicCombinatorics} for a complete treatment). 

Giving a step further, we define the rational function
\[
\alpha^P_\ell \defeq \dfrac{\phi^P}{\phi^{P \setminus \ell}}.
\]
It follows that $\phi^{P\setminus \ell}=\phi^{P\setminus m}$ if and only if $\alpha^P_\ell = \alpha^P_m$, and also equivalently, $\alpha^{P\setminus m}_\ell=\alpha^{P\setminus \ell}_m$. See \cite{coutinho2023strong} for the first appearance of this terminology, along with applications to study cospectrality in graphs. In this paper we will make use of this approach, so for later reference, we state this as a lemma.

\begin{lemma}\label{lem:cospectral_alpha} Let $\ell$ and $m$ be vertices in a weighted path $P$ with eigenvalues $\theta_0 > \cdots > \theta_d$. Then, the following are equivalent:
\begin{enumerate}[(a)]
    \item $\ell$ and $m$ are cospectral, that is, for all $s$, $|\braket{\ell}{\theta_s}| = |\braket{m}{\theta_s}|$;
    \item $\phi^{P\setminus \ell}=\phi^{P\setminus m}$;
    \item $\alpha^P_\ell=\alpha^P_m$;
    \item $\alpha^{P\setminus m}_\ell=\alpha^{P\setminus \ell}_m$.
\end{enumerate}
\end{lemma}


\section{Orthogonal polynomial sequences} \label{sec:orthogonal_polynomial_sequence}

The goal of this Section is to prove Theorem~\ref{thm:construction_OPS}, which provides a characterization of when two given arbitrary polynomials belong to a $d$--chain.

\subsection{Necessary conditions for an OPS}\label{sec:necessary_conditions}

We say that two polynomials with real zeros of different degrees (\emph{strongly}) \emph{interlace} if, in any (\emph{open}) \emph{closed} interval determined by two roots of the polynomial of lower degree, there is a root of the polynomial of higher degree, and if all the zeros of the polynomial of lower degree lie in the open interval determined by the smallest and largest zeros of the polynomial of higher degree. Note that if two polynomials strongly interlace, then all their zeros are simple.

First, notice that if $p_m$ and $p_{d+1}$ are in an orthogonal polynomial sequence, then $p_m$ and $p_{d+1}$ strongly interlace~\cite[Thm. 6.2]{chihara2011introduction}. This is a known result, but for completeness, we present a new proof using our techniques in Corollary~\ref{cor:strongly_interlace}. 

Denote by $\hat{p}_{d-m}$ the characteristic polynomial of the principal submatrix of $J$, defined in Equation~\eqref{eq:jacobi_matrix}, consisting of the last $d-m$ rows and columns. We also denote by $\overline{p}_{m-1}$  the characteristic polynomial of the principal submatrix of $J$ consisting of the rows and columns $\{d-m+1,\dots, d-1\}$.

We now present a result that is a direct consequence of the Christoffel-Darboux formula~\cite[Cor. 2.2]{GodsilAlgebraicCombinatorics} and establishes a connection between the polynomials $p_m$, $\hat{p}_m$, and $\overline{p}_{m-1}$.

\begin{theorem}\label{thm:cd_formula} Let $p_0,\dots, p_{d+1}$ be an orthogonal polynomial sequence. Then, for any $m$ in $\{0,\dots,d\}$, we have:

\begin{equation}\label{eq:cd1}\hat{p}_{d-m} p_d-\overline{p}_{m-1} p_{d+1} = \left(\prod_{t=m+1}^d\lambda_t^2\right)p_m,\end{equation}
\noindent which can also be rewritten as
\begin{equation}\label{eq:cd2}\dfrac{p_d}{p_{d+1}} = \left(\prod_{t=m+1}^d\lambda_t^2\right)\dfrac{p_m}{\hat{p}_{d-m} p_{d+1}} + \dfrac{\overline{p}_{m-1}}{\hat{p}_{d-m}}.\end{equation}
\end{theorem}
\begin{proof} Consider the $d$--chain as a weighted path $P$, and let $\phi^P$ denote its characteristic polynomial. Note that, by the Christoffel-Darboux formula~\cite[Cor. 2.2]{GodsilAlgebraicCombinatorics},

\[\phi^{P\setminus m}\phi^{P\setminus d}-\phi^{P\setminus \{m, d\}}\phi^P = \left(\prod_{t=m+1}^d\lambda_t^2\right)(\phi^{P\setminus \{m, \dots, d\}})^2.\]

Observe that $\phi^{P\setminus m} = p_m\hat{p}_{d-m}$, represents the characteristic polynomial of the matrix $J$ with the $m$-th row and column removed. This matrix is block-diagonal, containing the matrices corresponding to $p_m$ and $\hat{p}_{d-m}$. Similarly, $\phi^{P\setminus \{ m,d\}}=p_m\overline{p}_{m-1}$. Substituting these expressions into the formula above yields:
\[p_m\hat{p}_{d-m} p_d-p_m\overline{p}_{m-1} p_{d+1} = \left(\prod_{t=m+1}^d\lambda_t^2\right) p_m^2,\]
\noindent Finally, dividing both sides of the equation by $p_m$ results in the desired expression.
\end{proof}

\begin{corollary}\label{cor:equal_gcds}
     Let $p_0$, $p_1$, $\cdots$, $p_{d+1}$ be an orthogonal polynomial sequence. Then, $\gcd(p_m,p_{d+1})=\gcd(\hat{p}_{d-m},p_{d+1})=\gcd(p_m,\hat p_{d-m})$.
\end{corollary}
\begin{proof} 
Note that 
$\gcd(p_d, p_{d+1}) = 1$ and $\gcd(\overline{p}_{m-1}, \hat{p}_{d-m})=1$, since $p_d$ and $p_{d+1}$, as well as $\overline{p}_{m-1}$ and $\hat{p}_{d-m}$, are consecutive polynomials in an orthogonal polynomial sequence and therefore strongly interlace~\cite[Thm. 5.3]{chihara2011introduction}. This observation, together with Equation~\eqref{eq:cd1}, immediately implies the result.
\end{proof}

As a corollary, we obtain that two polynomials in an orthogonal polynomial sequence strongly interlace~\cite[Thm. 6.2]{chihara2011introduction}.

\begin{corollary}\label{cor:strongly_interlace}
 Let $p_0$, $p_1$, $\cdots$, $p_{d+1}$ be an orthogonal polynomial sequence. Then, $p_m$ and $p_{d+1}$ strongly interlace.
\end{corollary}
\begin{proof}
First, observe that for every $k$ in $\{0, \dots, d\}$, as $p_k$ and $p_{k+1}$ strongly interlace~\cite[Thm. 5.3]{chihara2011introduction}, the zeros of $p_k$ are in the open interval determined by the smallest and largest zeros of $p_{k+1}$. It follows that the zeros of $p_m$ are in the open interval determined by the smallest and largest zeros of $p_{d+1}$.

Next, note that $\phi^{P \setminus m} = p_m \hat{p}_{d-m}$ and $\phi^P = p_{d+1}$ interlace by the Cauchy Interlacing Theorem. By Corollary~\ref{cor:equal_gcds}, let $g$ be defined as $\gcd(p_m, p_{d+1}) = \gcd(\hat{p}_{d-m}, p_{d+1})$. It follows that the polynomials $p_m \hat{p}_{d-m}/g$ and $p_{d+1}/g$ interlace, as the same zeros are removed from both polynomials. In particular, we conclude that $p_m$ and $p_{d+1}/g$ interlace, and, as these polynomials are coprime, we obtain that they strongly interlace.
\end{proof}

\begin{corollary}\label{cor:partial_fraction} Let $p_0$, $p_1$, $\dots$, $p_{d+1}$ be an orthogonal polynomial sequence. Then,
\begin{equation*}\dfrac{p_d}{p_{d+1}}(x) = \displaystyle\sum_{s=0}^d \dfrac{\tau_s}{x-\theta_s}
\end{equation*}
\noindent where, for every $s$ in $\{0,\dots, d\}$,
\begin{equation}
\tau_s\defeq\dfrac{p_m}{\hat{p}_{d-m}}(\theta_s)\dfrac{\prod_{t=m+1}^d\lambda_t^2}{\prod_{r\neq s}(\theta_s-\theta_r)}+\displaystyle\lim_{x\to \theta_s}\dfrac{\overline{p}_{m-1}(x)(x-\theta_s)}{\hat{p}_{d-m}(x)}.\end{equation}
\end{corollary}
\begin{proof} First, observe that a partial fraction expansion exists for $p_d/p_{d+1}$~\cite[p. 29]{chihara2011introduction}. To determine the numerators, we simply use Equation~\eqref{eq:cd2}. Note that the first term on the RHS cannot be $\infty$, because by Corollary~\ref{cor:equal_gcds}, $\gcd(\hat{p}_{d-m},p_{d+1})=\gcd(\hat{p}_{d-m},p_m)$.
\end{proof}

Before we proceed to the next section, we observe some further properties of the coefficients $\tau_s$ in Corollary~\ref{cor:partial_fraction}. An immediate consequence is that $\tau_s > 0$ for every $s$ and $\sum_{s=0}^d \tau_s = 1$.

Let $I$ be the subset of indices $s$ such that $\theta_s$ is a common zero of $p_m$ and $p_{d+1}$. Since $\overline{p}_{m-1}$ and $\hat{p}_{d-m}$ strongly interlace, and by Corollary \ref{cor:equal_gcds}, we have
\[
\lim_{x \to \theta_s} \frac{\overline{p}_{m-1}(x)(x - \theta_s)}{\hat{p}_{d-m}(x)} \geq 0,
\]
for every $s$, with equality if and only if $s$ is not in $I$. Furthermore, observe that if $s$ is not in $I$, then, since $\tau_s > 0$, it follows that
\[
\frac{p_m}{\hat{p}_{d-m}}(\theta_s)\frac{\prod_{t=m+1}^d \lambda_t^2}{\prod_{r \neq s} (\theta_s - \theta_r)} > 0.
\]

It is also possible to prove that this expression is negative if $s$ is in $I$.


\subsection{Construction of orthogonal polynomial sequence}\label{sec:construction_OPS}

In this section, we present our first main result, which provides an algorithm that constructs an orthogonal polynomial sequence with two given polynomials $q_m$ and $q_{d+1}$ that strongly interlace. This is a known result when $m$ equals $d$ (see~\cite{VinetZhedanovHowTo} or~\cite[Lem. 5.1]{GodsilAlgebraicCombinatorics}). We manage to fully extend it for any $m$ by showing that strong interlacing is also sufficient.

\begin{theorem}\label{thm:construction_OPS}
Let $q_m$ and $q_{d+1}$ be two monic polynomials with distinct real zeros of degrees $m$ and $d + 1$, respectively, where $m < d + 1$. Then, there exists an orthogonal polynomial sequence $(p_k)_k$ with $p_m$ and $p_{d+1}$ equal to $q_m$ and $q_{d+1}$, respectively, if and only if $q_m$ and $q_{d+1}$ strongly interlace.
\end{theorem}

We note that our construction in the previous theorem produces all possible orthogonal polynomial sequences containing the polynomials $q_m$ and $q_{d+1}$. As our proof shows, there are $d-m$ degrees of freedom in the construction of such a sequence.

Assume that $q_{d+1} = \prod_{s=0}^d (x - \theta_s)$ and let $J$ be the set of indices $s$ such that $\theta_s$ is a zero of $q_m$. We will proceed backward from Section~\ref{sec:necessary_conditions}. We start with the following result.

\begin{proposition}\label{prop:q_hat} There exists a monic polynomial $\hat{q}_{d-m}$ with distinct real zeros and degree $d - m$, such that $\hat{q}_{d-m}(\theta_s) = 0$ for every $s$ in $J$, and
\[
\frac{q_m}{\hat{q}_{d-m}}(\theta_s) ~\frac{1}{\displaystyle\prod_{r \neq s}(\theta_s - \theta_r)}>0
\]
\noindent for $s$ not in $J$.
\end{proposition}
\begin{proof} Observe that, since $q_m$ and $q_{d+1}$ strongly interlace, there are exactly $d - m - |J|$ closed intervals determined by consecutive zeros of $q_{d+1}$ that do not contain a zero of $q_m$. Choose $d - m - |J|$ real numbers $\mu_1, \dots, \mu_{d - m - |J|}$ in the interior of these closed intervals, one for each interval. Let $\hat{q}_{d-m}$ be the polynomial
\[
(x - \mu_1) \cdots (x - \mu_{d - m - |J|}) \prod_{s \in J} (x - \theta_s).
\]
Note that $q_m \hat{q}_{d-m}$ has degree $d$ and interlaces $q_{d+1}$, therefore
\[
\lim_{x\to \theta_s}\dfrac{q_m \hat{q}_{d-m}}{q_{d+1}}(x)~(x-\theta_s)\geq 0,
\]
\noindent for every $s$, with equality if and only if $s$ is in $J$. If $s$ is not in $J$, then $\hat{q}_{d-m}(\theta_s)\neq 0$, and therefore
\[
\frac{q_m}{\hat{q}_{d-m}}(\theta_s) ~\frac{1}{\displaystyle\prod_{r \neq s}(\theta_s - \theta_r)} = \dfrac{1}{\hat{q}^2_{d-m}(\theta_s)}\lim_{x\to \theta_s}\dfrac{q_m \hat{q}_{d-m}}{q_{d+1}}(x)~(x-\theta_s)>0.
\]
\end{proof}

\begin{proposition}\label{prop:tau_s} Let $q_m$ and $\hat{q}_{d-m}$ as in the statement of Proposition~\ref{prop:q_hat}. Then, there exists $\Lambda > 0$ and $\rho_s \geq 0$ for $s$ in $\{0, \dots, d\}$ such that $\rho_s > 0$ if and only if $s$ is in $J$, and such that,
\[
\tau_s \defeq  \frac{q_m}{\hat{q}_{d-m}}(\theta_s)~ \frac{\Lambda}{\prod_{r \neq s} (\theta_s - \theta_r)} + \rho_s  > 0 \text{ for every $s$, and } \sum_{s=0}^d \tau_s = 1.
\]
\end{proposition}
\begin{proof} If $J$ is empty, we set $\rho_s = 0$ for every $s$. Since all the terms are positive according to Proposition~\ref{prop:q_hat}, we only need to choose $\Lambda$ as a normalization constant to satisfy the conditions. 

If $J$ is non-empty, we can first choose $\Lambda > 0$ sufficiently small and then, we choose $\rho_s > 0$ for $s$ in $J$ in such a way that all the conditions are satisfied. For instance, we can take $\Lambda > 0$ small such that,
\[
\left|\frac{q_m}{\hat{q}_{d-m}}(\theta_s)~ \frac{\Lambda}{\prod_{r \neq s} (\theta_s - \theta_r)}\right| < \dfrac{1}{2(d+1)},
\]
for every $s$, and then choose
\[
\rho_s \defeq \dfrac{1}{|J|}\left( 1 - \sum_{s=0}^d \frac{q_m}{\hat{q}_{d-m}}(\theta_s)~ \frac{\Lambda}{\prod_{r \neq s} (\theta_s - \theta_r)} \right)> \dfrac{1}{2|J|}>\dfrac{1}{2(d+1)},
\]
for $s$ in $J$ and $\rho_s=0$ if $s$ is not in $J$. Observe that in this case, by Proposition~\ref{prop:q_hat}, we clearly have $\sum_{s=0}^d \tau_s = 1$ and $\tau_s > 0$ for every $s$ in $\{0,\dots, d\}$.
\end{proof}

Now we are ready for the proof of the main result of this section.
    
\begin{proof}[Proof of Theorem~\ref{thm:construction_OPS}] 
Note that the necessity in the statement follows directly from Corollary~\ref{cor:strongly_interlace}. We will now prove the sufficiency part of the statement.

Consider the polynomial $\hat{q}_{d-m}$ defined in Proposition~\ref{prop:q_hat} and the values $\tau_s > 0$ for $s$ in $\{0, \dots, d\}$ defined in Proposition~\ref{prop:tau_s}.
Now, define the polynomial $q_d$ by
\[
\frac{q_d}{q_{d+1}}(x) = \sum_{s=0}^d \frac{\tau_s}{x - \theta_s}.
\]
Observe that $q_d$ is monic, has degree $d$, distinct real zeros, and strongly interlaces $q_{d+1}$. Thus, we may consider the orthogonal polynomial sequence $p_0 \equiv 1$, $p_1$, $\dots$, $p_{d+1}$ with $p_d := q_d$ and $p_{d+1} := q_{d+1}$. To finish the proof we just need to prove that $p_m$ is equal to $q_m$.

Let $I$ be the set of indices $s$ such that $\theta_s$ is a common zero of $p_m$ and $p_{d+1}$. Consider the polynomials $P_m$ and $\hat{P}_{d-m}$ obtained from $p_m$ and $\hat{p}_{d-m}$ by removing common zeros. Observe, that $\deg{P_m}=m-|I|$ and $\deg{\hat{P}_{d-m}}=d-m-|I|$. Similarly, consider the polynomials $Q_m$ and $\hat{Q}_{d-m}$, and observe that $\deg{Q_m}=m-|J|$ and $\deg{\hat{Q}_{d-m}}=d-m-|J|$.

Observe that, by Corollary~\ref{cor:partial_fraction} and Proposition~\ref{prop:tau_s}, if $s$ is not in $I \cup J$, then
\begin{align*}
\frac{P_m}{\hat{P}_{d-m}}(\theta_s)~ \frac{\prod_{t=m+1}^d \lambda_t^2}{\prod_{r \neq s} (\theta_s - \theta_r)} &= \lim_{x \to \theta_s} (x - \theta_s) \frac{p_d}{p_{d+1}}(x) \\
& = \lim_{x \to \theta_s} (x - \theta_s) \frac{q_d}{q_{d+1}}(x) \\ &= \frac{Q_m}{\hat{Q}_{d-m}}(\theta_s)~ \frac{\Lambda}{\prod_{r \neq s} (\theta_s - \theta_r)},
\end{align*}
which leads to the relation
\[
\left( \prod_{t=m+1}^d \lambda_t^2 \right) P_m(\theta_s) \hat{Q}_{d-m}(\theta_s) = \Lambda Q_m(\theta_s) \hat{P}_{d-m}(\theta_s),
\]
for every $s$ not in $I \cup J$. But, since $\deg(P_m \hat{Q}_{d-m})$ and $\deg(Q_m \hat{P}_{d-m})$ are both equal to $d - |I| - |J|$, we have
\[
\deg\left( \left( \prod_{t=m+1}^d \lambda_t^2 \right) P_m \hat{Q}_{d-m} - \Lambda Q_m \hat{P}_{d-m} \right) \leq d - |I| - |J| < |\{0, \dots, d\} \setminus (I \cup J)|.
\]
As a consequence,
\[
\left(\prod_{t=m+1}^d \lambda_t^2 \right) P_m \hat{Q}_{d-m} = \Lambda Q_m\hat{P}_{d-m}.
\]
Since the polynomials are monic, $\gcd(P_m, \hat{P}_{d-m}) = 1$ and $\gcd(Q_m, \hat{Q}_{d-m}) = 1$, it follows that $
\left( \prod_{t=m+1}^d \lambda_t^2 \right) = \Lambda$, $P_m = Q_m$ and $\hat{P}_{d-m} = \hat{Q}_{d-m}$. Thus, we conclude that $I = J$, $p_m = q_m$, and $\hat{p}_{d-m} = \hat{q}_{d-m}$, as we wanted to prove.
\end{proof}


\section{Cospectral vertices positions}\label{sec:cospectral_vertices_positions}

In this section, we study the possible positions of cospectral vertices in a $d$--chain. More precisely, we fully characterize the possible $0\leq \ell<m\leq d$ for which there exists some $d$--chain with vertices at positions $\ell$ and $m$ cospectral. 

\begin{theorem}\label{thm:cospectral_characterization}
Let $0\leq \ell<m\leq d$. There exists a $d$--chain in which $\ell$ and $m$ are cospectral if, and only if, $\ell<\frac{d}{2}<m$.
\end{theorem}

We proceed towards the proof of Theorem~\ref{thm:cospectral_characterization}. Our next result constrains the position of cospectral pairs of vertices in a chain and proves the necessity part of Theorem~\ref{thm:cospectral_characterization}.

\begin{proposition}\label{prop:restriction_cospectrality} Let $\ell<m$ be a cospectral pair of vertices in a $d$--chain. Then, $\ell<\dfrac{d}{2}<m$.
\end{proposition}
\begin{proof} We have by Lemma~\ref{lem:orthonormality} that $\ell$ and $m$ are cospectral if, and only if ${p_{d+1}\mid p_\ell^2-Cp_m^2}$, for some $C>0$. It follows that,

\[d+1=\deg(p_{d+1})\leq\deg(p_\ell^2-Cp_m^2)=2m,\]

\noindent and therefore $d/2<m$. In other words, we must have at least one vertex in the second half of the path. Now notice that we cannot have both vertices in the second half of the path, because in this case a reflection of the path would give us a path with two cospectral vertices in its first half.
\end{proof}

This result also shows that there are no triplets of cospectral vertices in chains, for in this case two of them would be on the same half of the path, which is forbidden by Proposition~\ref{prop:restriction_cospectrality}.

We proceed to prove the sufficiency part of Theorem~\ref{thm:cospectral_characterization}. Our next result is a direct consequence of the results of the previous sections.

\begin{proposition}\label{thm:construction_cospectral} Let $p_{d+1}$ be a monic polynomial with integer zeros $\theta_d<\cdots<\theta_0$, and let $\frac{d}{2}<m\leq d$. Then, the following are equivalent:
\begin{enumerate}[(a)]
    \item There exists a $d$--chain with spectrum $\{\theta_0, \dots, \theta_d\}$ in which the vertices $0$ and $m$ are cospectral.
    \item There exists a monic polynomial $p_m$ of degree $m$ with real zeros, which strongly interlaces $p_{d+1}$, such that $|p_m(\theta_s)|$ is constant for all $\theta_s$.
\end{enumerate}
\end{proposition}
\begin{proof} The implication (a) $\implies$ (b) is clear by  Lemma~\ref{lem:orthonormality} and Corollary~\ref{cor:strongly_interlace}. Therefore assume that $p_m$ is a polynomial with the properties as in (b). Since $p_m$ and $p_{d+1}$ are monic polynomials of degrees $m$ and $d+1$ that strongly interlace, Theorem~\ref{thm:construction_OPS} implies that there is an orthogonal polynomial sequence ${(p_k(x))_{k=0,\dots, d}}$ containing both. This orthogonal polynomial sequence satisfies a three-term recurrence relation~\cite{chihara2011introduction} $p_{k+1}(x):=(x-a_k)p_k-\lambda_k^2 p_{k-1}(x)$, for $k$ in $[d]$, where $p_0\equiv 1$ and $p_1(x)=x-a_0$. Using these coefficients $a_k$ and $\lambda_k$ we can then define a Jacobi matrix $J$, as in Equation~\eqref{eq:jacobi_matrix}. In this way, we obtain a $d$--chain with spectrum $\{\theta_0,\dots, \theta_d\}$ with these polynomials $p_m$ and $p_{d+1}$. As $|p_m(\theta_s)|$ is constant as $s$ varies in $\{0,\dots, d\}$, by Lemma~\ref{lem:orthonormality}, we obtain that $0$ and $m$ are cospectral in this $d$--chain. \end{proof}

As a consequence of Theorem~\ref{thm:construction_cospectral}, we can construct a chain with cospectral vertices $0$ and 
$m$, provided the condition in Proposition~\ref{prop:restriction_cospectrality} is satisfied.

\begin{proposition}\label{prop:existence_cospectrality} Let $d/2<m\leq d$. Then, a $d$--chain exists in which $0$ and $m$ form a cospectral pair.  
\end{proposition}
\begin{proof} Let $p_m$ be any monic polynomial with $m$ distinct real zeros. Fix an $\epsilon>0$ sufficiently small such that $p_m\pm \epsilon$ has distinct real zeros. Let $\rho_{2m-1}<\cdots <\rho_0$ be the union of the zero sets of $p_m-\epsilon$ and $p_m+\epsilon$. Note that the zeros of $p_m$ are in the open interval $(\rho_{2m-1},\rho_0)$ and that between two consecutive zeros of $p_m$ there is always an element of $\{\rho_2,\rho_4,\dots, \rho_{2m-2}\}$. For a given $m\leq d<2m$ let $S$ be the union of $\{\rho_0, \rho_{2m-1}\}\sqcup \{\rho_2,\rho_4,\dots, \rho_{2m-2}\}$ with a subset of $\{\rho_1,\rho_3,\dots, \rho_{2m-3}\}$ with $d-m$ elements. Note that $S$ has $d+1$ elements. Let $p_{d+1}$ be the monic polynomial of degree $d+1$ with zero set $S$, and observe that $p_{d+1}$ and $p_m$ strongly interlace and that $|p_m|$ is constant on $S$. Using these polynomials $p_m$ and $p_{d+1}$, by Theorem~\ref{thm:construction_cospectral}, we can obtain a $d$--chain with vertices $0$ and $m$ cospectral.
\end{proof}

The next result shows that given a chain with a cospectral pair of vertices, we can always construct a larger chain with a cospectral pair of vertices. The proof of this result uses techniques and results developed in~\cite{coutinho2023strong, coutinho2024no, coutinho2025spectrum}. Although the results in these works are stated for unweighted graphs, they also apply to weighted graphs with simple modifications.

\begin{proposition}\label{prop:extension_cospectrality} Let $\ell<m$ be a cospectral pair in a $d$--chain, and $k$ be a natural number. Then, there exists a $(d+2k)$--chain in which $\ell+k<m+k$ form a cospectral pair.
\end{proposition}
\begin{proof} It suffices to prove the statement for $k=1$, since iterating this case yields the general result. Let $\ell < m$ be a cospectral pair of vertices in a $d$--chain, which we view as a weighted path $P$. By Lemma~\ref{lem:cospectral_alpha}, we have $\alpha^{P\setminus m}_\ell = \alpha^{P\setminus \ell}_m$. We are going to construct a $(d+2)$--chain $\tilde{P}$ such that
\[
\alpha^{\tilde{P}\setminus m+1}_{\ell+1} = \alpha^{P\setminus m}_\ell - \frac{1}{x-u} = \alpha^{P\setminus \ell}_m - \frac{1}{x-u} = \alpha^{\tilde{P}\setminus \ell+1}_{m+1},
\]
for some appropriately chosen real number $u$. By Lemma~\ref{lem:cospectral_alpha} this will establish the result.

Assume that $P$ has vertex weights $(a_s)_{s=0,\dots,d}$ and edge weights $(\lambda_s)_{s=1,\dots,d}$, with associated orthogonal polynomial sequence $(p_s)_{s=0,\dots,d+1}$. By Equation~(4) in~\cite{coutinho2023strong}, we have
\[
\alpha^{P\setminus m}_\ell = \alpha^{P\setminus \{\ell-1, m\}}_\ell - \frac{\lambda_\ell^2}{\alpha^{P\setminus \ell}_{\ell-1}}, 
\; \text{and} \;
\alpha^{P\setminus \ell}_m = \alpha^{P\setminus \{m+1, \ell\}}_\ell - \frac{\lambda_{m+1}^2}{\alpha^{P\setminus m}_{m+1}},
\]
where we set $\lambda_0 = 0 = \lambda_{d+1}$ if $\ell=0$ or $m=d$. Moreover, we have
\[
\alpha^{P\setminus \ell}_{\ell-1} = \frac{\phi^{P\setminus \ell}}{\phi^{P\setminus \{\ell, \ell-1\}}} = \frac{p_\ell \hat{p}_{d-\ell}}{p_{\ell-1} \hat{p}_{d-\ell}} = \frac{p_\ell}{p_{\ell-1}}, \, \text{and similarly} \,
\alpha^{P\setminus m}_{m+1} = \frac{\hat{p}_{d-m}}{\hat{p}_{d-m-1}}.
\]

Since $p_{\ell}$ and $p_{\ell-1}$, and $\hat{p}_{d-m}$ and $\hat{p}_{d-m-1}$ strongly interlace, by Corollary~\ref{cor:partial_fraction} we have
\[
\frac{\lambda_\ell^2}{\alpha^{P\setminus \ell}_{\ell-1}} = \sum_{s=1}^\ell \frac{\mu_s}{x-r_s}, 
\quad \text{and} \quad
\frac{\lambda_{m+1}^2}{\alpha^{P\setminus m}_{m+1}} = \sum_{s=1}^{d-m} \frac{\tau_s}{x-t_s},
\]
where $r_1,\dots,r_\ell$ and $t_1,\dots,t_{d-m}$ are distinct real numbers, and $\mu_1,\dots,\mu_\ell$, $\tau_1,\dots,\tau_{d-m}$ are positive. 

Let $u$ be a real number different from $r_1,\dots,r_\ell$ and $t_1,\dots,t_{d-m}$. Define the monic polynomials $q_\ell$, $q_{\ell+1}$, $\hat{q}_{d-m}$, and $\hat{q}_{d+1-m}$ by
\[
\tilde{\lambda}_{\ell+1}^2 \frac{q_{\ell+1}}{q_\ell} = \sum_{s=1}^\ell \frac{\mu_s}{x-r_s} + \frac{1}{x-u}, 
\quad \text{and} \quad
\tilde{\lambda}_{m+2}^2 \frac{\hat{q}_{d+1-m}}{\hat{q}_{d-m}} = \sum_{s=1}^{d-m} \frac{\tau_s}{x-t_s} + \frac{1}{x-u},
\]
where $\tilde{\lambda}_{\ell+1}$ and $\tilde{\lambda}_{m+2}$ are positive constants.

As $q_{\ell+1}$ and $q_{\ell}$, and $\hat{q}_{d+1-m}$ and $\hat{q}_{d-m}$ strongly interlace (by Lemma~4 in~\cite{coutinho2024no}), there exist orthogonal polynomial sequences $(q_s)_{s=0,\dots,\ell+1}$ and $(\hat{q}_s)_{s=0,\dots,d+1-m}$ containing these polynomials and satisfying the three-term recurrences
\[
q_1(x) = x - \tilde{a}_0, \; 
q_{s+1}(x) = (x - \tilde{a}_s) q_s(x) - \tilde{\lambda}_s^2 q_{s-1}(x), \quad s \in [\ell],
\]
and
\[
\hat{q}_1(x) = x - \tilde{a}_{d+2}, \; 
\hat{q}_{s+1}(x) = (x - \tilde{a}_{d+2-s}) \hat{q}_s(x) - \tilde{\lambda}_{d+2-(s-1)}^2 \hat{q}_{s-1}(x), \quad s \in [d-m].
\]

For $s$ in $\{\ell +1, \dots, m+1\}$ define $\tilde{a}_s\defeq a_{s-1}$, and for $s$ in $\{\ell+2,\dots, m+1\}$ define $\tilde{\lambda}_s\defeq \lambda_{s-1}$. Let $\tilde{P}$ have vertex weights $(\tilde{a}_s)_{s=0,\dots,d+2}$ and edge weights $(\tilde{\lambda}_s)_{s=1,\dots,d+2}$ with associated orthogonal polynomial sequence $(\tilde{p}_k)_{k=0,\dots, d+3}$.

By Equation~(4) in~\cite{coutinho2023strong}, 
\[
\alpha^{\tilde{P}\setminus m+1}_{\ell+1} = \alpha^{\tilde{P}\setminus \{\ell, m+1\}}_{\ell+1} - \frac{\tilde{\lambda}_{\ell+1}^2}{\alpha^{\tilde{P}\setminus \ell+1}_{\ell}}.
\]

Now, note that,
\[
\alpha^{\tilde{P}\setminus \{\ell, m+1\}}_{\ell+1} 
= \frac{\phi^{\tilde{P}\setminus \{\ell, m+1\}}}{\phi^{\tilde{P}\setminus \{\ell, \ell+1, m+1\}}} 
= \frac{\phi^{P\setminus \{\ell-1, m\}}}{\phi^{P\setminus \{\ell-1, \ell, m\}}} 
= \alpha^{P\setminus \{\ell-1,m\}}_{\ell},
\]
where in the second equality we used that $\tilde{a}_s = a_{s-1}$ for $s \in \{\ell+1, \dots, m+1\}$ and $\tilde{\lambda}_s = \lambda_{s-1}$ for $s \in \{\ell+2, \dots, m+1\}$. Moreover,
\[
\frac{\tilde{\lambda}_{\ell+1}^2}{\alpha^{\tilde{P}\setminus \ell+1}_{\ell}}
= \tilde{\lambda}_{\ell+1}^2 \frac{\tilde{p}_{\ell}}{\tilde{p}_{\ell+1}} 
= \tilde{\lambda}_{\ell+1}^2 \frac{q_\ell}{q_{\ell+1}} = \sum_{s=1}^\ell \frac{\mu_s}{x - r_s} + \frac{1}{x-u} 
= \frac{\lambda_\ell^2}{\alpha^{P\setminus \ell}_{\ell-1}} + \frac{1}{x-u}.
\]

Therefore,
\[
\alpha^{\tilde{P}\setminus m+1}_{\ell+1} = \alpha^{P\setminus m}_\ell - \frac{1}{x-u}.
\] 
The proof that 
\[
\alpha^{\tilde{P}\setminus \ell+1}_{m+1} = \alpha^{P\setminus \ell}_m - \frac{1}{x-u}
\] 
is analogous.
\end{proof}

We are now ready for the proof of the main result of this section.

\begin{proof}[Proof of Theorem~\ref{thm:cospectral_characterization}] By Proposition~\ref{prop:restriction_cospectrality} we know that if $\ell$ and $m$ are cospectral in a $d$--chain, then $\ell<\dfrac{d}{2}<m$. For the other direction, consider $\ell<\dfrac{d}{2}<m\leq d$. Note that if we do a reflection on the $d$--chain, then we may deal with the pair $d-m<d/2<d-l$. Therefore, we can also assume without loss, by possibly doing a reflection, that $m+\ell\leq d$. In this case, we have $0< \frac{d-2\ell}{2}<m-\ell\leq d-2l$, and thus, by Proposition~\ref{prop:existence_cospectrality}, we know that there is a $(d-2\ell)$--chain with $0$ and $m-\ell$ cospectral. We conclude the result by using Proposition~\ref{prop:extension_cospectrality} with this chain for $k=\ell$.
\end{proof}


\section{Periodic cospectral vertices}\label{sec:periodic_cospectral_vertices}

Having fully characterized the positions of cospectral vertices in $d$--chains in the previous section, we turn to the study of \emph{periodic} cospectral vertices. We say vertex at $\ell$ is periodic if Equation~\eqref{eq:pst} holds with $\ell = m$, which, as before, is equivalent to having, without loss of generality,
\begin{align}
    \exp(\ii (2\pi) J) \ket \ell = \ket \ell. \label{eq:periodic}
\end{align}
It is straightforward to notice that if perfect state transfer happens between $\ell$ and $m$ at time $\tau$, then both are periodic at time $2\tau$, so in order to study vertices involved in state transfer, we may address the necessary condition of both vertices being cospectral and periodic.

We consider $d$--chains with $\ell=0$ and $\tfrac{d}{2}<m\leq d$ cospectral and periodic. By Lemma~\ref{lem:orthonormality}, this is equivalent to requiring $|p_m(\theta_s)| = C|p_0(\theta_s)|=C$ and $\theta_s \in \mathbb{Z}$ for every $s$.

For such a pair of periodic cospectral vertices, denote by $S_+$ and $S_-$ the sets of eigenvalues $\theta_s$ with $p_m(\theta_s) > 0$ and $p_m(\theta_s) < 0$, respectively, so that $S_+ \sqcup S_- = \{\theta_s\}_{s=0,\dots,d}$. Note that $p_m$ is uniquely determined by $S_+$ and $S_-$. Indeed, by Lagrange interpolation there exists a unique polynomial of degree at most $d$ that takes the value $1$ on $S_+$ and $-1$ on $S_-$, and the monic normalization of this polynomial is precisely $p_m$.

Since $S_+$ and $S_-$ uniquely determine $p_m$, if we fix the spectrum $\{\theta_s\}_{s=0,\dots,d}$, it suffices to study the admissible configurations of these two sets to determine the possible periodic cospectral vertices positions. The next result provides a restriction on such configurations. In particular, it shows that, in the ordering $\theta_d < \cdots < \theta_0$, there are never three consecutive elements all belonging to $S_+$ or $S_-$.

\begin{proposition}\label{prop:zeros_position} Let $\tau_{m-1}<\cdots <\tau_0$ and $\theta_d<\cdots <\theta_0$ be the zeros of $p_m$ and $p_{d+1}$, respectively. Then, for each $i$ in $\{0,\dots, m-2\}$, the interval $(\tau_{i+1}, \tau_i)$ contains either one or two zeros $\theta_j$ of $p_{d+1}$, which are in $S_+$ or $S_-$ depending on whether $i$ is odd or even, respectively. Moreover, the intervals $(-\infty, \tau_{m-1})$ and $(\tau_0, \infty)$ contain only $\theta_d$ and $\theta_0$, respectively. In particular, $\theta_0\in S_+$ and $\theta_d\in S_+$ if and only if $m$ is even.
\end{proposition}
\begin{proof} Consider $i$ in $\{0, \dots, m-2\}$ and the corresponding interval $(\tau_{i+1}, \tau_i)$. Since $p_m$ is monic, its sign in this interval depends on the parity of $i$. Thus, any zero $\theta_j$ of $p_{d+1}$ in $(\tau_{i+1}, \tau_i)$ belongs to $S_+$ or $S_-$ depending on whether $i$ is odd or even, respectively.

Because $p_m$ and $p_{d+1}$ strongly interlace, there is at least one zero $\theta_j$ in $(\tau_{i+1}, \tau_i)$. We claim that there are at most two. Suppose, for contradiction, that $\theta_{j+1} < \theta_j < \theta_{j-1}$ are all contained in $(\tau_{i+1}, \tau_i)$. Then $p_m(\theta_{j+1}) = p_m(\theta_j) = p_m(\theta_{j-1}) = \pm C$, so by Rolle’s Theorem, there exist distinct zeros $\rho_-$ and $\rho_+$ of $p_m'$ in $(\theta_{j+1}, \theta_j)$ and $(\theta_j, \theta_{j-1})$, respectively. Since $p_m$ is real-rooted, $p_m$ and $p_m'$ interlace, there must be a zero $\tau$ of $p_m$ in $[\rho_-, \rho_+]$, and hence in $(\tau_{i+1}, \tau_i)$, which is a contradiction. This proves our claim.

A similar argument shows that the intervals $(-\infty, \tau_{m-1})$ and $(\tau_0, \infty)$ contain only $\theta_d$ and $\theta_0$, respectively.
\end{proof} 

By combining Proposition~\ref{thm:construction_cospectral} with Proposition~\ref{prop:zeros_position}, we can now show that the most challenging case for constructing periodic cospectral vertices occurs when the vertices are at positions $0$ and $\lceil \frac{d+1}{2} \rceil$.

\begin{theorem}\label{thm:cospectrality_ordering_cases} Let $\frac{d}{2}<m\leq d$ be such that there exists a $d$--chain with vertices $0$ and $m$ cospectral and periodic. Then, for every $m\leq d'\leq d$ there exists a $d'$--chain with vertices $0$ and $m$ cospectral and periodic.
\end{theorem}
\begin{proof} By Proposition~\ref{thm:construction_cospectral}, let $p_m$ and $p_{d+1}$ be the monic polynomials associated with a $d$--chain where the vertices $0$ and $m$ are cospectral and periodic, and denote by $S\defeq S_+\sqcup S_-$ the set of zeros of $p_{d+1}$. Observe that $S$ consists of integers.

We claim that, for any integer $d'$ with $m\leq d'\leq d$, it is possible to remove $d-d'$ elements from $S$ to obtain a subset $S'\subseteq S$ such that the monic polynomial $q_{d'+1}$, whose zero set is $S'$, strongly interlaces $p_m$.

Let $\tau_{m-1}<\cdots <\tau_0$ be the zeros of $p_m$. By Proposition~\ref{prop:zeros_position}, there exists a subset $E\subseteq\{0,\dots,m-2\}$ consisting of the indices $i$ for which the interval $(\tau_{i+1},\tau_i)$ contains exactly two zeros of $p_{d+1}$. Note that $E$ has exactly $d-m$ elements. Thus we may select a subset $R\subseteq E$ of size $d-d'$. For each $i$ in $R$, we choose one $\theta_{k_i}$ of the two zeros of $p_{d+1}$ in $(\tau_{i+1},\tau_i)$, and define $S'\defeq S\setminus\{\theta_{k_i}\}_{i\in R}$. By construction, $S'$ has exactly $d'+1$ distinct elements. We now show that the monic polynomial $q_{d'+1}$, whose zero set is $S'$, strongly interlaces $p_m$. Indeed, Proposition~\ref{prop:zeros_position} guarantees that both $\theta_d$ and $\theta_0$ are in $S'$, so all zeros of $p_m$ lie in the open interval determined by the smallest and largest zeros of $q_{d'+1}$. Moreover, for each $i$ in $\{0, \dots, m-2\}$, the interval $(\tau_{i+1}, \tau_i)$ contains at least one zero of $q_{d'+1}$ by construction. Hence, $p_m$ and $q_{d'+1}$ strongly interlace.

Finally, by Proposition~\ref{thm:construction_cospectral}, there exists a $d'$--chain in which $0$ and $m$ are cospectral, corresponding to the polynomials $p_m$ and $q_{d'+1}$. Since the set $S'$ consists of integers, it follows that $0$ and $m$ are also periodic in this $d'$--chain.
\end{proof}

Note that, by Theorem~\ref{thm:cospectrality_ordering_cases}, the hardest case is constructing chains with periodic cospectral vertices at positions $0$ and $\lceil \frac{d+1}{2}\rceil$. One of our paper highlights is to show that constructing these chains is related to a classical problem in number theory.

\subsection{Prouhet--Tarry--Escott problem}\label{sec:PTE_problem}

We denote by $PTE_n$ the problem of finding two distinct multisets of integers $E=\{e_1, \dots, e_n\}$ and $F=\{f_1, \dots, f_n\}$ such that
\[
e_1^k+\cdots +e_n^k=f_1^k+\cdots +f_n^k
\]
for every $k$ in $\{1,\dots, n-1\}$. These are known as the ideal solutions to the Prouhet--Tarry--Escott problem, a classical problem with a history dating back over two centuries. Solutions for $PTE_n$ are known only for $n \leq 10$ and $n = 12$. For further information on this problem and its history, we recommend~\cite{borwein2002prouhet,caley2013prouhet,coppersmith2024ideal}.

As we will see later in Corollary~\ref{cor:PTE_interlacing}, the multisets $E$ and $F$ are disjoint and each element in $E$ or $F$ appears at most twice. We define the problem $PTE_n^0$ just the same as $PTE_n$ but with the additional requirement that no element is repeated in $E\cup F$, that is, $E$ and $F$ are disjoint sets. We also define $PTE_n^1$ as the variant of $PTE_n$ where at most one element in $E \cup F$ appears twice, and all others are distinct. Clearly, a solution to $PTE_n^0$ is also a solution to $PTE_n^1$, which in turn is a solution to $PTE_n$.

In this section we prove the following result.

\begin{theorem}\label{thm:PTE_equivalence}Let $d$ be a natural number. If $d$ is odd, then there exists a $d$--chain with periodic cospectral vertices at positions $0$ and $\left\lceil\frac{d+1}{2}\right\rceil$ if and only if there exists a solution to $PTE_{\left\lceil\frac{d+1}{2}\right\rceil}^0$. If $d$ is even, then such a $d$--chain exists if and only if there exists a solution to $PTE_{\left\lceil\frac{d+1}{2}\right\rceil}^1$.
\end{theorem}

By combining Theorems~\ref{thm:cospectrality_ordering_cases} and~\ref{thm:PTE_equivalence}, we obtain the following result, which shows that determining all possible positions of pairs of periodic cospectral vertices $\ell = 0$ and $\frac{d}{2}<m\leq d$ in $d$--chains is nearly equivalent to solving the Prouhet--Tarry--Escott problem.

\begin{corollary}\label{cor:PTE_almost_equivalent} Determining all possible positions of pairs of periodic cospectral vertices $0$ and $\frac{d}{2}<m\leq d$ in $d$--chains requires solving the problem $PTE_n^0$ for every natural number $n$. Conversely, a positive solution to $PTE_n^0$ for each $n$ yields a complete description of such positions.
\end{corollary}
\begin{proof} By Theorem~\ref{thm:PTE_equivalence}, determining whether there exists a $d$--chain with $d$ odd such that $0$ and $\lceil\frac{d+1}{2}\rceil$ are periodic and cospectral requires solving $PTE_{\lceil\frac{d+1}{2}\rceil}^0$. Therefore, by varying $d$ over all odd natural numbers, we must solve $PTE_n^0$ for every natural number $n$. This proves one direction of the statement.

Conversely, assume that there is a positive solution to $PTE_n^0$ for every natural number $n$. In particular, we also obtain a positive solution to $PTE_n^1$ for every $n$. By Theorem~\ref{thm:PTE_equivalence}, for every natural number $d$ we can then obtain a $d$--chain with $0$ and $\lceil\frac{d+1}{2}\rceil$ periodic and cospectral. But then, by Theorem~~\ref{thm:cospectrality_ordering_cases}, for every $d$ and $\frac{d}{2}<m\leq d$ we can construct a $d$--chain with vertices $0$ and $m$ periodic and cospectral.
\end{proof}

Moreover, Theorems~\ref{thm:cospectrality_ordering_cases} and~\ref{thm:PTE_equivalence}, together with the results of Section~\ref{sec:construction_OPS}, provide an algorithm to convert solutions of $PTE_n^0$ or $PTE_n^1$ into chains with pairs of periodic cospectral vertices. Note that further number theoretic requirements are needed in order to guarantee these pairs admit perfect state transfer between them.

To prove Theorem~\ref{thm:PTE_equivalence} we use the following reformulation of $PTE_n$ in terms of polynomials. For a multiset $E=\{e_1,\dots, e_n\}$, let $p_E$ be the polynomial $p_E(x)\defeq \prod_{j=1}^n (x-e_j)$.

\begin{theorem}[Theorem 1 in~\cite{borwein2002prouhet}]\label{thm:PTE_polynomials} Let $E$ and $F$ be two distinct multisets of integers with $n$ elements. Then, $E$ and $F$ form a solution to $PTE_n$ if, and only if, $p_E-p_F$ is a constant.
\end{theorem}

The following result is a direct consequence of Theorem~\ref{thm:PTE_polynomials} and is referred to as the interlacing theorem for the Prouhet--Tarry--Escott problem (see~\cite[Thm. 2.1.3]{caley2013prouhet}, or Exercise E2 in~\cite{borwein2002prouhet}). In our setting, this result is an analogue of Proposition~\ref{prop:zeros_position}.

\begin{corollary}[Thm. 2.1.3 in~\cite{caley2013prouhet}]\label{cor:PTE_interlacing} If $E=\{e_1,\dots, e_n\}$ and $F=\{f_1,\dots, f_n\}$ are multisets that form a solution to $PTE_n$ with $e_1\leq e_2\leq \cdots \leq e_n$ and $f_1\leq f_2\leq \cdots \leq f_n$, then 
\[
e_1<f_1\leq f_2<e_2\leq e_3< \cdots <e_{n-1}\leq e_n<f_n,\text{if $n$ is odd,}
\]
\[
e_1<f_1\leq f_2<e_2\leq e_3< \cdots <e_{n-2}\leq e_{n-1}<f_{n-1}\leq f_n<e_n,\text{if $n$ is even,}
\]
\noindent where we assume without loss of generality that $e_1<f_1$.    
\end{corollary}

We are now ready for the proof of Theorem~\ref{thm:PTE_equivalence}.

\begin{proof}[Proof of Theorem~\ref{thm:PTE_equivalence}] \

$(\Longrightarrow)$

Assume that $d \geq 1$, the case $d = 1$ is trivial. Consider a $d$--chain in which the vertices $0$ and $\left\lceil\frac{d+1}{2}\right\rceil$ are cospectral and periodic, and the corresponding sets $S_+$ and $S_-$. Note that $S_+$ and $S_-$ are disjoint sets of integers. 

We first claim that both $|S_+|$ and $|S_-|$ equal $\left\lceil\frac{d+1}{2}\right\rceil$ when $d$ is odd, and that $\{|S_+|, |S_-|\} = \left\{\left\lceil\frac{d+1}{2}\right\rceil, \left\lfloor\frac{d+1}{2}\right\rfloor\right\}$ when $d$ is even. 

\begin{itemize}
    \item[] Since $|S_+|+|S_-|=d+1$, it suffices to show that both $|S_+|$ and $|S_-|$ are at most $\left\lceil\frac{d+1}{2}\right\rceil$ to establish the claim. Let $\tau_{\lceil \frac{d-1}{2} \rceil} < \cdots < \tau_0$ be the zeros of $p_{\lceil \frac{d+1}{2} \rceil}$. By Proposition~\ref{prop:zeros_position}, all the zeros $\theta_{d-1} < \cdots < \theta_1$ lie in the interval $(\tau_{\lceil \frac{d-1}{2} \rceil}, \tau_0)$, and those that are in $S_+$ lie in the intervals $(\tau_{i+1}, \tau_i)$ for odd indices $i$ in the set $\{0, 1, \dots, \lceil \frac{d+1}{2} \rceil - 2\}$. Again, by Proposition~\ref{prop:zeros_position}, each such interval contains at most two elements from $S_+$. There are exactly $\lfloor \frac{1}{2} (\lceil \frac{d+1}{2} \rceil - 2) \rfloor$ such odd indices $i$, so the number of elements in $S_+ \cap (\tau_{\lceil \frac{d-1}{2} \rceil}, \tau_0)$ is at most $2 \lfloor \frac{1}{2} (\lceil \frac{d+1}{2} \rceil - 2) \rfloor = 2 \lfloor \frac{d}{4} \rfloor$. Note that, by Proposition~\ref{prop:zeros_position}, $\theta_0 \in S_+$, and $\theta_d \in S_+$ if and only if $\lceil \frac{d+1}{2} \rceil$ is even, which occurs precisely when $d \equiv 2$ or $3 \pmod{4}$. Thus, $|S_+|$ is at most $2 \lfloor \frac{d}{4} \rfloor + 1$ if $d \equiv 0$ or $1 \pmod{4}$, and at most $2 \lfloor \frac{d}{4} \rfloor + 2$ if $d \equiv 2$ or $3 \pmod{4}$. In any case, we conclude that $|S_+| \leq \lceil \frac{d+1}{2} \rceil$. Moreover, by Proposition~\ref{prop:zeros_position}, among the zeros $\theta_{d-1} < \cdots < \theta_1$, at most $\lceil \frac{d-1}{2} \rceil$ are in $S_-$. Since $\theta_0 \notin S_-$ and possibly $\theta_d\in S_-$, we obtain that $|S_-| \leq \lceil \frac{d-1}{2} \rceil + 1 = \lceil \frac{d+1}{2} \rceil$. This proves our claim.
\end{itemize}

Now consider the polynomials $p_{S_+}$ and $p_{S_-}$, and note that $p_{d+1}=p_{S_+}p_{S_-}$. From our earlier discussion, we know that if $d$ is odd, then both $p_{S_+}$ and $p_{S_-}$ have degree $\lceil\frac{d+1}{2}\rceil$, while if $d$ is even, their degrees are $\lceil\frac{d+1}{2} \rceil$ and $\lfloor\frac{d+1}{2} \rfloor$ in some order. Since $|p_{\lceil\frac{d+1}{2}\rceil}|$ is constant on the zeros of $p_{d+1}$, there exists a constant $C$ such that $p_{S_+}$ divides $p_{\lceil\frac{d+1}{2}\rceil}-C$ and $p_{S_-}$ divides $p_{\lceil\frac{d+1}{2}\rceil}+C$. We now analyze the two cases: $d$ odd and $d$ even.

Assume that $d$ is odd. Since $p_{S_+}$, $p_{S_-}$, and $p_{\lceil\frac{d+1}{2}\rceil} \pm C$ are all monic of the same degree $\lceil\frac{d+1}{2}\rceil$, it follows that $p_{S_+} = p_{\lceil\frac{d+1}{2}\rceil} - C$ and $p_{S_-} = p_{\lceil\frac{d+1}{2}\rceil} + C$. Therefore, we have $p_{S_+}-p_{S_-} = -2C$. Hence by Theorem~\ref{thm:PTE_polynomials}, since $S_+$ and $S_-$ are sets of integers, they form a solution to $PTE_{\lceil\frac{d+1}{2}\rceil}$. Furthermore, because $S_+$ and $S_-$ are disjoint sets, this gives a solution to $PTE_{\lceil\frac{d+1}{2}\rceil}^0$.

Now, assume that $d$ is even. Without loss of generality, suppose that $|S_+|=\left\lceil\frac{d+1}{2}\right\rceil$ and $|S_-|=\left\lfloor\frac{d+1}{2}\right\rfloor$,  the proof for the other case being analogous. In this case, we have $p_{S_+}=p_{\lceil\frac{d+1}{2}\rceil}-C$ and $(x-\xi)p_{S_-}=p_{\lceil\frac{d+1}{2}\rceil}+C$ for some real number $\xi$. Therefore, $p_{S_+}-(x-\xi)p_{S_-}=-2C$. From this expression, it follows that the sum of the elements in $S_+$ equals $\xi$ plus the sum of the elements in $S_-$. Since $S_+$ and $S_-$ are integer sets, it follows that $\xi$ is an integer. By Theorem~\ref{thm:PTE_polynomials}, it then follows that the multisets $S_+$ and $S_- \cup \{\xi\}$ form a solution to $PTE_{\lceil\frac{d+1}{2}\rceil}$. Moreover, as $S_+$ and $S_-$ are disjoint sets, the only possible repeated element in the multiset $S_+\cup S_-\cup \{\xi\}$ is $\xi$. Therefore, $S_+$ and $S_- \cup \{\xi\}$ form a solution to $PTE_{\lceil \frac{d+1}{2} \rceil}^1$. This completes the proof of one direction of the statement.

$(\Longleftarrow)$

Let $d$ be odd and assume that there exists a solution formed by sets $E=\{e_1,\dots, e_{\lceil \frac{d+1}{2} \rceil}\}$ and $F=\{f_1,\dots, f_{\lceil \frac{d+1}{2} \rceil}\}$ for $PTE_{\lceil \frac{d+1}{2} \rceil}^0$. Consider the polynomials $p_E$ and $p_F$. By Theorem~\ref{thm:PTE_polynomials}, the difference $p_E - p_F$ is a non-zero constant. Define  $p_{d+1} \defeq p_E p_F$ and  $p_{\lceil\frac{d+1}{2}\rceil}\defeq \frac{p_E + p_F}{2}$. Observe that $p_{d+1}$ and $p_{\lceil\frac{d+1}{2}\rceil}$ are monic polynomials of degrees $d+1$ and $\lceil\frac{d+1}{2}\rceil$, respectively, and that $p_{d+1}$ has distinct integer zeros $E\sqcup F$. Also note that $|p_{\lceil\frac{d+1}{2}\rceil}|$ is constant on the zeros of $p_{d+1}$. 

We claim that $p_{\lceil\frac{d+1}{2}\rceil}$ has distinct real zeros and strongly interlaces $p_{d+1}$. To see this, consider Corollary~\ref{cor:PTE_interlacing} and assume without loss of generality that $e_1 < f_1$. By Corollary~\ref{cor:PTE_interlacing}, for each odd index $i$ in $\{1, \dots, \lceil\frac{d+1}{2}\rceil\}$, we have $e_i < f_i$, and by the definition of $p_{\lceil\frac{d+1}{2}\rceil}$, it follows that $p_{\lceil\frac{d+1}{2}\rceil}(e_i) = -p_{\lceil\frac{d+1}{2}\rceil}(f_i) \neq 0$. Therefore, by continuity, $p_{\lceil\frac{d+1}{2}\rceil}$ has a zero in $(e_i, f_i)$. Similarly, for each even $j$ in $\{1, \dots, \lceil\frac{d+1}{2}\rceil\}$, $p_{\lceil\frac{d+1}{2}\rceil}$ has a zero in $(f_j, e_j)$. Hence, $p_{\lceil\frac{d+1}{2}\rceil}$ has at least $\lceil\frac{d+1}{2}\rceil$ distinct real zeros. Since its degree is also $\lceil\frac{d+1}{2}\rceil$, these are all its zeros. Note that this argument also shows that all the zeros of $p_{\lceil\frac{d+1}{2}\rceil}$ lie strictly between the smallest and largest zeros of $p_{d+1}$. Moreover, between any two consecutive zeros of $p_{\lceil\frac{d+1}{2}\rceil}$, there are two consecutive elements from either $E$ or $F$, i.e. at least two zeros of $p_{d+1}$. Therefore, $p_{\lceil\frac{d+1}{2}\rceil}$ and $p_{d+1}$ strongly interlace. By applying Proposition~\ref{thm:construction_cospectral} to $p_{d+1}$ and $p_{\lceil\frac{d+1}{2}\rceil}$, we conclude that there exists a $d$–chain in which $0$ and $\lceil\frac{d+1}{2}\rceil$ are periodic and cospectral.

Now, assume that $d$ is even and that there exists a solution given by multisets $E$ and $F$ for $PTE_n^1$. As before, consider the polynomials $p_E$ and $p_F$ and define $q_{d+2} \defeq p_E p_F$ and $p_{\lceil\frac{d+1}{2}\rceil} \defeq \frac{p_E + p_F}{2}$. Observe that $q_{d+2}$ and $p_{\lceil\frac{d+1}{2}\rceil}$ are monic polynomials of degrees $d+2$ and $\lceil\frac{d+1}{2}\rceil$, respectively, and that $q_{d+2}$ has zero set $E \cup F$. Also note that $|p_{\lceil\frac{d+1}{2}\rceil}|$ is constant on the zeros of $q_{d+2}$. Using the same reasoning as before, together with Corollary~\ref{cor:PTE_interlacing}, we can show that $p_{\lceil\frac{d+1}{2}\rceil}$ has real zeros and strongly interlaces $q_{d+2}$. Note that the multiset $E \cup F$ may contain an element $\xi$ that appears twice. In this case, we define $p_{d+1}(x) \defeq \frac{q_{d+2}(x)}{x - \xi}$. If no element appears twice in $E \cup F$, we choose $\xi$ to be any element in $E \cup F$ that is not in $\{e_1, f_1, e_{\lceil\frac{d+1}{2}\rceil}, f_{\lceil\frac{d+1}{2}\rceil}\}$ and define $p_{d+1}(x) \defeq \frac{q_{d+2}(x)}{x - \xi}$. In either case, it follows from Corollary~\ref{cor:PTE_interlacing} and our previous argument that $p_{d+1}$ has distinct integer zeros and strongly interlaces $p_{\lceil\frac{d+1}{2}\rceil}$. Finally, by applying Proposition~\ref{thm:construction_cospectral} to $p_{d+1}$ and $p_{\lceil\frac{d+1}{2}\rceil}$, we conclude that there exists a $d$–chain in which $0$ and $\lceil\frac{d+1}{2}\rceil$ are periodic and cospectral.
\end{proof}


\section{Perfect State Transfer}\label{sec:PST}

In this section, we apply the theory developed in the previous sections to study $d$--chains with perfect state transfer between vertices $\ell = 0$ and $\tfrac{d}{2} < m \leq d$. We prove that, apart from trivial cases, there is no perfect state transfer between the first position and the one immediately after the middle of the chain. In Section~\ref{sec:PST_examples}, we also present new examples of positions admitting perfect state transfer.

The next result is the analogue of Proposition~\ref{thm:construction_cospectral} for perfect state transfer.

\begin{theorem}\label{thm:construction_pst} Let $p_{d+1}$ be a monic polynomial with integer zeros $\theta_d<\cdots<\theta_0$, and let $\frac{d}{2} < m \leq d$. Then, the following are equivalent:
\begin{enumerate}[(a)]
    \item There exists a $d$--chain with spectrum $\{\theta_0, \dots, \theta_d\}$ that admits perfect state transfer between vertices $0$ and $m$.
    \item There exists a monic polynomial $p_m$ of degree $m$ with real zeros, which strongly interlaces $p_{d+1}$, such that $p_m(\theta_s)=(-1)^{\theta_0-\theta_s} C$ for some constant $C>0$ and every $\theta_s$.
\end{enumerate}
\end{theorem}
\begin{proof} This result follows directly from Theorem~\ref{thm:PST_characterization_paths} and Proposition~\ref{thm:construction_cospectral}.
\end{proof}

We remark that if in the ordering $\theta_d < \cdots < \theta_0$ the eigenvalues occur in alternating blocks of even or odd integers of length one or two (as it should be by Proposition~\ref{prop:zeros_position}), then the strong interlacing of the polynomials $p_m$ and $p_{d+1}$ in item $(b)$ of Theorem~\ref{thm:construction_pst} follows automatically.

The next result is an analogue of Theorem~\ref{thm:cospectrality_ordering_cases} and suggests that the most challenging case is constructing $d$--chains with perfect state transfer between $0$ and $\lceil \frac{d+1}{2}\rceil$, precisely the situation in Theorem~\ref{thm:PTE_equivalence}.

\begin{theorem}\label{thm:pst_ordering_cases} Let $m\leq d$ be such that there exists a $d$--chain with perfect state transfer between  vertices $0$ and $m$. Then, for any $m\leq d'\leq d$, there exists a $d'$--chain with perfect state transfer between $0$ and $m$.
\end{theorem}
\begin{proof} The proof of this result is analogous to the proof of Theorem~\ref{thm:cospectrality_ordering_cases} and uses Theorem~\ref{thm:construction_pst}.
\end{proof}

Unlike the case of periodic and cospectral vertices in Theorem~\ref{thm:PTE_equivalence}, where the existence of $d$--chains corresponds to certain solutions of the Prouhet--Tarry--Escott problem, perfect state transfer between vertices $0$ and $\lceil \frac{d+1}{2} \rceil$ in a $d$--chain can actually be ruled out for $d \geq 4$. This shows that even for weighted $d$--chains perfect state transfer imposes stricter conditions than periodicity and cospectrality alone.

\begin{theorem}\label{thm:no_pst_half} If $d \geq 4$, then there is no perfect state transfer between vertices $0$ and $\lceil\frac{d+1}{2}\rceil$ in a $d$--chain.
\end{theorem}
\begin{proof} Assume that $d$ is such that there exists a $d$–chain with perfect state transfer between vertices $0$ and $\lceil\frac{d+1}{2}\rceil$. By Theorem~\ref{thm:construction_pst}, $S_+$ and $S_-$ are the sets of even and odd zeros of $p_{d+1}$, respectively.

From the proof of Theorem~\ref{thm:PTE_equivalence}, we know that: if $d$ is odd, then $S_+$ and $S_-$ form a solution to $PTE_{\lceil \frac{d+1}{2}\rceil}^0$; if $d$ is even, then there exists $\xi \in \mathbb{Z}$ such that either ${\xi} \cup S_+$ and $S_-$, or ${\xi} \cup S_-$ and $S_+$ form a solution to $PTE_{\lceil \frac{d+1}{2}\rceil}^1$.

In either case, we obtain a solution with multisets $E$ and $F$ to $PTE_{\lceil \frac{d+1}{2}\rceil}^1$ where we may assume, without loss of generality, that either: $E$ consists entirely of even integers, and $F$ contains at most one even number; $F$ consists entirely of odd integers, and $E$ contains at most one odd number. Consider the polynomials $p_E$ and $p_F$, and note that, by Theorem~\ref{thm:PTE_polynomials}, we have $p_E(x)-p_F(x)=C$ for some nonzero integer $C$.

Now, by a result of Kleiman~\cite{kleiman278note} (see also~\cite[p.~4]{coppersmith2024ideal}), if $G=\{g_1,\dots,g_n\}$ and $H =\{h_1, \dots, h_n\}$ form a solution to $PTE_n$, then 
\[(n-1)! \text{ divides } \prod_{i\in [n]}g_i - \prod_{j\in[n]}h_j = p_G(0)-p_H(0).\] 
Applying this to our setting, it follows that if $\lceil\frac{d+1}{2}\rceil -1\geq 2$, i.e., if $d\geq 4$, then $2$ divides $C$. 

Hence, if $d\geq 4$, $p_E(x) \equiv p_F(x)\pmod{2}$, and so $p_E$ and $p_F$ have the same factorization in $\mathbb{F}_2[x]$. But $p_E(x)$ has $0$ as a root with multiplicity at least $\lceil\frac{d+1}{2}\rceil -1\geq 2$ while $p_F(x)$ has $0$ as a root with multiplicity at most $1$ in $\Fds_2[x]$. Therefore, $p_E$ and $p_F$ cannot have the same factorization in $\Fds_2[x]$. We conclude that $d<4$.
\end{proof}

\subsection{Perfect State Transfer Examples}\label{sec:PST_examples}

In this section, we present examples of $d$--chains that exhibit perfect state transfer between vertices $0$ and $\frac{d}{2} < m \leq d$. We begin by providing examples that illustrate a converse to Theorem~\ref{thm:no_pst_half}.

Observe that for $d = 1$ and $d = 2$, the paths $P_2$ and $P_3$, with vertex weights equal to $0$ and edge weights equal to $1$, are examples of $d$--chains with perfect state transfer between vertices $0$ and $\lceil \frac{d+1}{2} \rceil$. For $d=3$, we have the following example.

\begin{example}\label{ex:pst_3} Let $p_4\defeq (x^2-1)(x^2-4)$ be the monic polynomial with zero set $\{\pm 1, \pm 2\}$. If we define $p_2 \defeq x^2 - \frac{5}{2}$, then $p_2(\theta_s) = (-1)^{2-\theta_s} \frac{3}{2}$ for all $\theta_s\in \{\pm 1, \pm 2\}$ . Observe that $p_2$ has real roots and strongly interlaces $p_4$. By Theorem~\ref{thm:construction_pst}, the pair $p_4$ and $p_2$ provides an example of a $3$--chain with perfect state transfer between vertices $0$ and $2$.
\end{example}

Since Theorem~\ref{thm:no_pst_half} rules out perfect state transfer between $0$ and $\lceil \tfrac{d+1}{2} \rceil$ for all $d \geq 4$, the next possible pair is $0$ and $\lceil\frac{d+3}{2}\rceil$. When $d$ is odd, a sufficient (though not necessary) condition for the existence of $d$--chains with perfect state transfer between these positions is the existence of a solution to $PTE_{\frac{d+3}{2}}^0$ in which both sets contain exactly one odd number. The following example illustrates this construction.

\begin{example} A non-exhaustive search among the known solutions to $PTE_5$ yields the following two solutions to $PTE_5^0$, each with exactly one odd number in each set:
\[
E_1\defeq\{-8 , -4 , 0 , 8 , 9 \}, \quad F_1\defeq\{-7 , -6 , 2 , 6 , 10\},
\]
\[
E_2\defeq\{-55 , -24 , -6 , 32 , 58\}, \quad F_2\defeq\{-52 , -34 , 9 , 22 , 60\}.
\]
Let $\tilde{E}_1$ and $\tilde{F}_1$ be $E_1$ and $F_1$ divided by two. Define the monic polynomial 
\[
p_8(x) \defeq x(x-1)(x+2)(x-3)(x+3)(x-4)(x+4)(x-5),
\] 
whose zero set is exactly the integers in $\tilde{E}_1 \sqcup \tilde{F}_1$, and set 
\[
p_5(x) \defeq \frac{1}{2}\bigl(p_{\tilde{E}_1}(x) + p_{\tilde{F}_1}(x)\bigr) = x^5 - \frac{5}{2}x^4 - 25x^3 + 40x^2 + 144x - \frac{315}{4}.
\] 
Then $p_8$ and $p_5$ strongly interlace, and for every root $\theta_s$ of $p_8$ we have
\[
p_5(\theta_s) = (-1)^{5-\theta_s} \frac{315}{4}.
\] 
By Theorem~\ref{thm:construction_pst}, this gives a $7$--chain with perfect state transfer between $0$ and $5$. Moreover, by Theorem~\ref{thm:pst_ordering_cases}, it also gives a $6$--chain with perfect state transfer between $0$ and $5$.
\end{example}


\section{Questions}\

Our first question asks for a complete characterization of the positions of pairs of periodic cospectral vertices (or admitting perfect state transfer) in $d$--chains.

\begin{question} Consider a $d$--chain. What are the possible positions of pairs of periodic cospectral vertices $\ell<m$? For which of these pairs is perfect state transfer possible?
\end{question}

We observe that Proposition~\ref{thm:construction_cospectral} and Theorem~\ref{thm:construction_pst} address only the case $\ell=0$. This is because Theorem~\ref{thm:construction_OPS} characterizes when two polynomials belong to an orthogonal polynomial sequence, whereas no analogous characterization is known for three polynomials. Moreover, even in this restricted setting with $\ell=0$, number-theoretic obstructions arise, as shown by Theorems~\ref{thm:PTE_equivalence} and~\ref{thm:no_pst_half}. This motivates the following problem.

\begin{problem} For every sufficiently large $d$, construct $d$--chains such that the vertices $0$ and $m_d$, with $\frac{d}{2} < m_d \leq d$, are cospectral and periodic (or exhibit perfect state transfer), and such that $\lim_{d \to \infty} \frac{m_d}{d}$ is as small as possible.
\end{problem}


\section*{Acknowledgements}

The authors thank Alastair Kay for conversations about the topic of this paper a while ago. We acknowledge the financial support of FAPEMIG, CNPq and CAPES.

\bibliographystyle{plain}
\bibliography{references.bib}

	
\end{document}